\documentclass{article}
\usepackage{romp}
\usepackage[english]{babel}
\usepackage{amsmath}
\usepackage{amssymb}
\usepackage{amsthm}

\usepackage{graphicx}
\usepackage{bm}
\usepackage[cp1250]{inputenc}
\usepackage{tabularx,rotating,lscape}
\usepackage{float}

\renewcommand{\[}{\begin{equation}}
\renewcommand{\]}{\end{equation}}

\newcommand{\pd}{\partial}

\newcommand{\R}{\mathbb{R}}
\newcommand{\C}{\mathbb{C}}
\renewcommand{\AA}{\mathcal{A}}
\newcommand{\ba}{\mathbf{a}}
\newcommand{\bA}{\mathbf{A}}
\newcommand{\BB}{\mathcal{B}}

\newcommand{\re}{\mathrm{e}}
\newcommand{\HH}{\mathcal{H}}
\newcommand{\rmi}{\mathrm{i}}

\newcommand{\bK}{\mathbf{K}}
\newcommand{\bp}{\mathbf{p}}
\newcommand{\rR}{\mathcal{R}}
\newcommand{\bx}{\mathbf{x}}

\newcommand{\ve}{\varepsilon}

\newcommand{\sign}{\mathop{\rm sign}}

\newcommand{\mN}{\mathbb{N}}

\renewcommand{\[}{\begin{equation}}
\renewcommand{\]}{\end{equation}}

\renewcommand{\AA}{\mathcal{A}}

\newtheorem{remark}{Remark}

\begin{document}

\title{Dynamics of an electron confined to a ``hybrid plane'' and interacting with a magnetic field}

\author{Raffaele Carlone\\Universit\`a dell'Insubria, via Valleggio, 22100 Como, Italy\\
e-mail:  raffaele.carlone@me.com \\[2ex] Pavel Exner\\
Nuclear Physics Institute, Czech Academy of Sciences, 25068 \v{R}e\v{z} near Prague, \\
and Doppler Institute, Czech Technical University,
B\v{r}ehov\'{a}~7, 11519~Prague, Czechia\\e-mail: exner@ujf.cas.cz
}

\maketitle

\begin{abstract}
We discuss spectral and resonance properties of a Hamiltonian
describing motion of an electron moving on a ``hybrid surface'' consisting on a halfline attached by its endpoint to a plane
under influence of a constant magnetic field which interacts with its spin through a Rashba-type term.
\end{abstract}

\noindent\emph{Keywords:} quantum transport, contact interaction, spin-orbit coupling

\setcounter{equation}{0}
\section{Introduction}

\noindent The study of quantum particles confined to manifolds of a mixed dimensionality has a long history starting from \cite{ES87}. Recently this problem  attracted a new interest connected with possible influence of external fields and/or internal degrees of freedom. In this paper we continue this line of research and study the dynamics of a charged particle with spin $\frac12$, having in mind an electron, which moves on a ``hybrid surface'' consisting of a halfline attached by its endpoint to a plane. In the plane the electron interacts with a constant magnetic field with orientation perpendicular to it, and with its own spin via a spin-orbit interaction --- in the present work we suppose that the latter is of the form due to Rashba.

The coupling at the contact point between the halfline and the
plane can be chosen in different ways to make the resulting
Hamiltonian self-adjoint. In this sense it may include a point
interaction; such a motion in the plane alone was analyzed in
\cite{AP05}; the hybrid plane without the magnetic field was
discussed \cite{ES07}.

The construction of the Hamiltonian follows the usual pattern using the theory of self-adjoint extensions. As a starting information we use free evolution on the disconnected parts of the configuration space. On the one hand it is the halfline equipped with the Laplacian with Neumann boundary condition. This is, of course, simple. More complicated is the other part describing a particle with spin moving in the plane under the effect of a constant magnetic field and Rashba interaction --- here we can use the results of the papers \cite{BGP07} and \cite{BGP007} which we briefly recall in the next section.

After this preliminary we construct in Section~3 Hamiltonians
describing the coupled system; we write the appropriate
generalized boundary conditions and derive the corresponding Green function. In Section~4 we study properties of these Hamiltonians for a spin-independent coupling. We analyze scattering of a particle travelling on the halfline finding the corresponding reflection amplitude. Furthermore, we find the spectrum of the Hamiltonian as well as the resonances the origin of which are the Landau levels in the decoupled plane.

\setcounter{equation}{0}
\section{A preliminary: motion in the plane}

In this section we collect some results about the motion in the plane needed in the following; proofs and details can be found in~\cite{BGP07}. We consider a charged two-dimensional particle\footnote{To hedge us against a terminological objection we hasten to add that we have in mind a two-dimensional \emph{model} of a real-world particle which proved to be extremely efficient, in particular, as a description of a rare electron gas confined to a very thin layer.} of spin $\frac12$ under the influence of a uniform magnetic field of intensity $B$ orthogonal to the plane --- for definiteness we suppose that $B>0$. Let $\bA$ be the corresponding magnetic vector potential, $B=\frac{\partial A_y}{\partial x} -\frac{\partial A_x}{\partial y}$; we employ the symmetric gauge putting $A(x,y)= \left(\frac12{By}, -\frac12{Bx}\right)$. The Hilbert state space is
\begin{equation*}
\HH_\mathrm{plane}=L^2(\R^2,\C^2)
\end{equation*}
If we take into account the Rashba spin-orbit interaction, the
motion of the particle in the plane is described by the following Hamiltonian
\[
\Hat H_{R}=\Hat H_0+\displaystyle\frac{\alpha_{R}}{\hbar}\,
\Hat U_{R}+\displaystyle\frac{g_*}{2}\,\mu_B B\sigma_z
\label{HamRas}\]
where
\begin{itemize}
\item $\Hat H_0=\displaystyle\frac{1}{2m_*}\bm{\Pi}^2\sigma_0$
with $\Pi_j:=-i\hbar \pd_j-\frac{e}{c}\,A_j$,
$\:j=x,y$, where $\sigma_0$ is the $2\times2$ identity matrix,
\item $\Hat U_{R}=\sigma_x\Pi_y-\sigma_y \Pi_x
$, with the standard notation for the  Pauli matrices,
$\sigma_x=\left(
\begin{array}{cc} 0&1\\
                1&0\\
\end{array}
\right)$, $\sigma_y=\left(
\begin{array}{cc} 0&-\rmi\\
                \rmi&0\\
\end{array}
\right)$, $\sigma_z=\left(
\begin{array}{cc} 1&0\\
                0&-1\\
\end{array}
\right)$,
\end{itemize}
and $\mu_B\equiv \displaystyle\frac{|e|\hbar}{2m_ec}$ is the Bohr magneton. Furthermore, $m_e$ and $m^*$ are the electron mass and its effective mass, respectively, $g_*$ is the effective $g$-factor, and $\alpha_{R}$ is the real-valued Rashba constant.

The values of the units are not important in the following hence we shall use mostly dimensionless coordinates introducing the following notation: the coupling $\varkappa_{R}:=
\displaystyle\frac{m_*\alpha_{R}}{\hbar^2}$,  the magnetic flux quantum $\Phi_0:=\displaystyle\frac{2\pi \hbar c}{e}$,
furthermore, $b:=\displaystyle\frac{2\pi}{\Phi_0}B$ and
$\ba:=\displaystyle{\frac{2\pi} {\Phi_0}}\,\bA$ or
$\ba=\left(\frac12{by},-\frac12{bx}\right)$. Setting now
$\gamma:=-\displaystyle\frac{g_*}{2}\,\displaystyle\frac{m_*}{m_e}$
we can rewrite the above Hamiltonian as
\begin{eqnarray*}
&H_{R}=K^2\,\sigma_0+2\varkappa_{R} U_{R}+\gamma b \sigma_z
\end{eqnarray*}
where $\bK:=\frac{1}{\hbar}\,\bp-\ba\,$ and $\,U_{R}=\sigma_x K_y-\sigma_y K_x$.

As usual the properties of such a Hamiltonian are encoded in its resolvent. The latter is known explicitly and setting
\[\beta_{R}:=\frac{\gamma+1}{2\varkappa_{R}}b,\quad
\eta_{R}:= \sqrt{z+\varkappa_{R}^2+\beta_{R}^2}\,,\quad
\zeta^{\pm}_{R}(b):=(\eta_{R}\pm\varkappa_{R})^2+ b -\beta^2_{R}
\]
it is possible to write an explicit formula for the Green function components
\[\label{greensimp}
G_{R}(\bx,\bx';z)=\left(\begin{array}{cc}
G^{11}_{R}(\bx,\bx';z) & G^{12}_{R}(\bx,\bx';z)\\
G^{21}_{R}(\bx,\bx';z) & G^{22}_{R}(\bx,\bx';z)
\end{array}\right)
\]
where $\bx\equiv(x,y)$, $\bx'\equiv(x',y')$, and
\begin{eqnarray*}
G^{11}_{R}(\bx,\bx';z)= \frac{\beta_{R}-\varkappa_{R}}{2\eta_{R}}
\Big(G_0\big(\bx,\bx';\zeta_{R}^-(b)\big)-
G_0\big(\bx,\bx';\zeta^+_{R}(b)\big)\Big)\\+
 \frac{1}{2}\Big(
G_0\big(\bx,\bx';\zeta_{R}^-( b)\big)+
G_0\big(\bx,\bx';\zeta^+_{R}(b)\big)\Big),\\
G^{22}_{R}(\bx,\bx';z)= -\frac{\beta_{R}+\varkappa_{R}}{2\eta_{R}}
\Big(G_0\big(\bx,\bx';\zeta_{R}^-(-b)\big)-
G_0\big(\bx,\bx';\zeta^+_{R}(-b)\big)\Big) \\+ \frac{1}{2}\Big(
G_0\big(\bx,\bx';\zeta_{R}^-( -b)\big)+
G_0\big(\bx,\bx';\zeta^+_{R}(-b)\big)\Big)\,,
\end{eqnarray*}
while the off-diagonal elements are
\begin{eqnarray*}
\lefteqn{G^{12}_{R}(\bx,\bx';z)}
\\ && =|b|\,\Big((x-x')-\rmi(y-y')\Big)\,\bigg(\frac{\sign b -1}{2}\,
\Big[G_0\big(\bx,\bx';\zeta_{R}^-(-b)\big) \\ && -
G_0\big(\bx,\bx';\zeta^+_{R}(-b)\big)\Big]
+\Big[F_0\big(\bx,\bx';\zeta_{R}^-(-b)\big)-
F_0\big(\bx,\bx';\zeta^+_{R}(-b)\big)\Big] \bigg)
\end{eqnarray*}
and $G^{21}_{R}(\bx,\bx';z)=\overline{G^{12}_{R}(\bx',\bx;\bar
z)}$. Here we have used the notation
\begin{eqnarray*}
\lefteqn{ F_0(\bx,\bx';z):=\displaystyle
\frac{1}{4\pi}\,\Big(\frac{z}{2|b|}-\frac{1}{2}\Big)
\Gamma\big(\frac{1}{2}-\frac{z}{2|b|}\big)}\\ &&\times
\exp\Big(\frac{\rmi
b}{2}(\bx\wedge\bx')-\frac{|b|}{4}(\bx-\bx')^2\Big)\,\Psi
\Big(\frac{3}{2}-\frac{z}{2|b|}\,,2;\frac{|b|}{2}(\bx-\bx')^2\Big).
\end{eqnarray*}
and
\begin{eqnarray*}
\lefteqn{G_0(\bx,\bx';z):=\frac{1}{4\pi}\,
\Gamma\big(\frac{1}{2}-\frac{z}{2|b|}\big)}\\ && \times
\exp\Big(\frac{\rmi
b}{2}(\bx\wedge\bx')-\frac{|b|}{4}(\bx-\bx')^2\Big)\,\Psi
\Big(\frac{1}{2}-\frac{z}{2|b|}\,,1;\frac{|b|}{2}(\bx-\bx')^2\Big),
\end{eqnarray*}
where $\Psi$ is the confluent hypergeometric function and $\Gamma$ is the Euler gamma function. As usual in these situations it is necessary to know also the renormalized Green function with the diagonal singularity removed given by
\begin{equation*}
G^{\mathrm{ren}}_R(z):=\lim_{\bx'\to\bx}\Big[
G_R(\bx,\bx';z)-S(\bx,\bx') \Big]\end{equation*}
where
$S(\bx,\bx';z):=-\displaystyle\frac{1}{2\pi}\log|\bx-\bx'|\sigma_0$.
If we put
\begin{equation}\label{qdef}
U(z):=\lim_{\bx\to \bx'}\Big(G_0(\bx,\bx';z)+\frac{1}{2\pi}\log |\bx-\bx'|\Big)=
-\frac{1}{4\pi}\Big(\psi\big(\frac{1}{2}-\frac{z}{2|b|}\big)-2\psi(1)+\log\frac{|b|}{2}\Big),
\end{equation}
then
\begin{eqnarray*}
G_R^{\mathrm{ren}}(z) &\!=\!& \left(\begin{array}{cc}G_{R,1}^{\mathrm{ren}}(z)&0\\
0&G_{R,2}^{\mathrm{ren}}(z)\end{array}\right)\\
&\!=\!&\frac{\beta_{R}\sigma_z-\varkappa_{R}} {2\eta_{R}}\,
\left(\begin{array}{cc} U\big(\zeta_{R}^-(b)\big)-
U\big(\zeta^+_{R}(b)\big)&0\\
0&U\big(\zeta_{R}^-(-b)\big)-
U\big(\zeta^+_{R}(-b)\big)\end{array}\right)\\
&&+\frac{1}{2}\left(\begin{array}{cc} U(\zeta_{R}^-( b))+
U(\zeta^+_{R}(b))&0\\
0&U(\zeta_{R}^-(-b))+
U\big(\zeta^+_{R}(-b)\big)\\
\end{array}\right)\,,
\end{eqnarray*}
where
\begin{equation}
\begin{split}\label{gg1}
G_{R,1}^{\mathrm{ren}}(z):=&-\frac{1}{4\pi}\left(\frac{1}{2}
+\frac{\beta_{R}-\varkappa_r}{2\sqrt{z+\beta_{R}^{2}+\varkappa_r^{2}}}\right)\psi\left(\frac{\beta_{R}^{2}
-(\sqrt{z+\beta_{R}^{2}+\varkappa_r^{2}}-\varkappa_r)^{2}}{2b}\right)\\
&-\frac{1}{4\pi}\left(\frac{1}{2}-\frac{\beta_{R}-\varkappa_r}{2\sqrt{z+\beta_{R}^{2}
+\varkappa_r^{2}}}\right)\psi\left(\frac{\beta_{R}^{2}-(\sqrt{z+\beta_{R}^{2}+\varkappa_r^{2}}+\varkappa_r)^{2}}{2b}\right)\\
&+\frac{1}{4\pi}\left(2\psi(1)-\log\left(\frac{b}{2}\right)\right)
\end{split}
\end{equation}
and
\begin{equation}
\begin{split}\label{g2}
G_{R,2}^{\mathrm{ren}}(z):=
&-\frac{1}{4\pi}\left(\frac{1}{2}-\frac{\beta_{R}+\varkappa_r}{2\sqrt{z+\beta_{R}^{2}
+\varkappa_r^{2}}}\right)\psi\left(\frac{2b+\beta_{R}^{2}-(\sqrt{z+\beta_{R}^{2}
+\varkappa_r^{2}}-\varkappa_r)^{2}}{2b}\right)\\
&-\frac{1}{4\pi}\left(\frac{1}{2}+\frac{\beta_{R}+\varkappa_r}{2\sqrt{z+\beta_{R}^{2}
+\varkappa_r^{2}}}\right)\psi\left(\frac{2b+\beta_{R}^{2}-(\sqrt{z+\beta_{R}^{2}+\varkappa_r^{2}}+\varkappa_r)^{2}}{2b}\right)\\
&+\frac{1}{4\pi}\left(2\psi(1)-\log\left(\frac{b}{2}\right)\right)
\end{split}
\end{equation}
The symbol $\psi(z)$ here means the digamma function, known to be a meromorphic function of $z$ with no branch cut discontinuities and with simple poles at $z=0,-1,-2,....$, defined by
\begin{equation*}
\psi(z)=-\gamma+(z-1)\sum_{n=0}^{\infty}\frac{1}{(n+1)(z+n)}\,.
\end{equation*}

The spectrum of the magnetic Rashba Hamiltonian has been derived in \cite{BGP07}. It consists of infinitely degenerate eigenvalues which are natural to call modified Landau levels,
\begin{equation}\begin{split}
&\sigma_p (H_{R})=\{\varepsilon^{\pm}(n,s):\,n\in\mN,\,s=\pm1\}\\
&\varepsilon^{\pm}(n,s)=|b|(2n+1- s\sign b)\pm
2\varkappa_{R}\sqrt{\beta_{R}^2+|b|\big(2n+1- s\sign b\big)}\,;
\end{split}
\end{equation}
one can recover this result by inspecting the resolvent
singularities in (\ref{greensimp}).

\setcounter{equation}{0}
\section{Motion in the hybrid plane}

After this preliminary we come to the problem described in the
introduction and suppose that the configuration space consists of a plane described above to which a halfline lead is attached; without loss of generality we place the junction to the origin of coordinates in the plane. The construction is analogous to the non-magnetic case \cite{ES07} the difference being in the plane Hamiltonian component, nevertheless, we describe it in sufficient detail in order to make the present paper self-contained.

\subsection{The halfline-plane coupling}

The lead Hilbert space is $\HH_\mathrm{lead}= L^2(\R_+,\C^2)$, and the whole state space of the system is the consequently the orthogonal sum $\HH:= \HH_\mathrm{lead} \oplus
\HH_\mathrm{plane}$. In other words, the wave functions are of the form $\Psi = \{ \psi_\mathrm{lead},
\psi_\mathrm{plane}\}^\mathrm{T}$ with each of the components
being a $2\times 1$ column. The construction starts from the
decoupled operator $H^0:=H_\mathrm{lead} \oplus H_R$ where the
first component is the Laplacian on the halfline $H_\mathrm{lead} \psi_\mathrm{lead} = -\psi''_\mathrm{lead}$ with Neumann boundary condition at the endpoint\footnote{There is no need to consider the effect of the magnetic field on the halfline even if the lead is not perpendicular to the plane because one can always remove it by a simple gauge transformation.}, while $H_R$ is the magnetic Rashba Hamiltonian discussed in the previous section. Following the method of \cite{ES87} we restrict the operator $H^0$ to functions which vanish in the vicinity of the junction obtaining thus a symmetric operator of deficiency indices $(4,4)$. In the second step we construct its self-adjoint extensions which are regarded as admissible Hamiltonians.

The extensions can be characterized in various ways. There is a general scheme coming from von Neumann theory, however, it is by far more practical to employ boundary conditions when dealing with a problem of the present type. The boundary values on the halfline are the simply columns $\psi_\mathrm{lead}(0+)$ and $\psi'_\mathrm{lead} (0+)$. On the other hand, in the plane we have to use generalized ones. The functions in the domain of the restriction have a logarithmic singularity at the origin and the generalized boundary values $L_j(\psi_\mathrm{plane}),\, \mbox{j=0,1,}\,$ appear as
coefficients in the corresponding expansion,
 \[ \label{gen_bv}
\psi_\mathrm{plane}(\bx) = -\frac{1}{2\pi}\,
L_0(\psi_\mathrm{plane})\,\ln |\bx| + L_1(\psi_\mathrm{plane}) +
o(|\bx|)\,.
 \]
 where
 \begin{equation*}
 \begin{split}
 &L_0(\psi_\mathrm{plane})=\lim_{|\bx|\to
 0}\frac{\psi_\mathrm{plane}(|\bx|)}{\ln{|\bx|}} \\[.5em]
 &L_1(\psi_\mathrm{plane})=\lim_{|\bx|\to 0}\left[\psi_\mathrm{plane}(|\bx|)
 -\ln(|\bx|)L_0(\psi_\mathrm{plane}(|\bx|))\right]
 \end{split}
 \end{equation*}
 Using these boundary values we can write the sought boundary conditions as
 \[ \label{bc}
\begin{array}{rcl}
\psi'_\mathrm{lead} (0+) &=& A\psi_\mathrm{lead} (0+) +
C^*L_0(\psi_\mathrm{plane})\,, \\ [.3em] L_1(\psi_\mathrm{plane})
&=& C\psi_\mathrm{lead} (0+) + DL_0(\psi_\mathrm{plane})\,,
\end{array}
 \]
where $A,D,C$ are $2\times 2$ matrices, the first two of them
Hermitian, so the matrix $\AA:= {A\;C^* \choose C\;D}$ depends of sixteen real parameters as the deficiency indices suggest. One can check easily that the corresponding boundary form vanishes under the condition (\ref{bc}), which means that each fixed $\AA$ gives rise to a self-adjoint extension $H_\AA$ of the restricted operator.

It is worth noting that the above boundary conditions are
generic but do not cover all the extensions leaving out cases
when the matrix $\AA$ is singular; this flaw can be mended in
the standard way \cite{KS, AP05} if one replaces
(\ref{bc}) by the symmetrized form of the relation,
 \[ \label{bc2}
\AA {\psi_\mathrm{lead}(0+) \choose L_0(\psi_\mathrm{plane})} +
\BB {\psi'_\mathrm{lead}(0+) \choose L_1(\psi_\mathrm{plane})}
=0\,,
 \]
where $\AA,\BB$ are matrices such that $(\AA|\BB)$ has rank four and $\AA\BB^*$ is Hermitean. We will restrict ourselves, however, to the generic case $\BB=-I$ expressed by (\ref{bc}) in the following; the same is true for the alternative form of the b.c. mentioned below.

The way in which the parameter matrix is chosen depends on
physical properties of the coupling between the lead and the
plane. In particular, diagonal $A,D,C$ correspond to the situation when the junction does not couple the spin states, and moreover, scalar matrices describe a spin-independent coupling. It is obvious that the lead and the plane are decoupled if $\AA$ is block-diagonal, i.e. $C=0$. A na\"{\i}ve interpretation of the conditions \eqref{bc} is that $C$ is responsible for the coupling while $A$ and $D$ are point perturbations at the two components of the configuration space, respectively.

\begin{remark}
{\rm An attentive reader may wonder about the relation between the boundary conditions and the geometry of the problem, that is, the angle between the halfline and the plane. Likewise, one may ask whether the coupling could be influenced by the magnetic field. The above analysis gives no answer to these questions; it only guarantees self-adjointness of the Hamiltonian, or in other words, conservation of the probability current through the junction. A natural approach would be to consider a ``fat'' hybrid plane in analogy with the analogous problem for networks -- cf.~\cite{EP09} and references therein -- and to analyze the limiting behavior as its thickness tends to zero. This question is open and by far not easy.}
\end{remark}

\subsection{The Green function}

As usual properties of an operator are encoded in its resolvent, hence our next task is to find the latter for the above constructed self-adjoint extensions. A suitable tool to do that is Krein's formula \cite{AGHH, BGMP02} which allows us to find the sought resolvent starting from Green's function of the decoupled system which has a block-diagonal form,
 \[ \label{freeG}
G^0(x,x';\bx,\bx';z)=\left(\begin{array}{cc}
G_\mathrm{lead}(x,x';z) & \mathbf{0}_{2} \\
\mathbf{0}_{2} & G_{R}(\bx,\bx';z) \end{array} \right)\,,
 \]
where $\mathbf{0}_{2}$ is the $2\times 2$  null matrix,
$G_{R}(\bx,\bx';z)$ is given by (\ref{greensimp}) and
 $$
 G_\mathrm{lead}(x,x';z)= \frac{\rmi}{\sqrt{z}}\, \cos
 \sqrt{z}x_<\: \re^{-\rmi\sqrt{z}x_>}\: \sigma_0
 $$
with the conventional notation, $x_<:=\min\{x,x'\},\:
x_>:=\max\{x,x'\}$, since we have assumed Neumann boundary
condition at the halfline endpoint. We introduce Krein's function $Q(z)$, which is an analytic $4\times4$-matrix valued function of the spectral parameter $z$, as the diagonal values of the kernel, with the above described renormalization in the planar part, specifically
 \[ \label{kreinf}
 Q(z):= \left(\begin{array}{cc}
 \frac{\rmi}{\sqrt{z}}\: \sigma_0 & \mathbf{ 0}_{2} \\
 \mathbf{0}_{2} & G^{\mathrm{ren}}_R(z) \end{array} \right)\,.
 \]
The full Green function is obtained by a finite-rank perturbation of the free one. It is convenient to rewrite the conditions \eqref{bc} using a modified basis in the boundary value space: instead of the vectors employed above we take
 $$
 \tilde\Gamma_1\psi:= {-\psi'_\mathrm{lead}(0+) \choose
 L_0(\psi_\mathrm{plane})}\,, \quad \tilde\Gamma_2\psi:=
 {\psi_\mathrm{lead}(0+) \choose L_1(\psi_\mathrm{plane})}\,.
 $$
One can check easily that they satisfy $\tilde\AA\tilde\Gamma_1
\psi + \tilde\BB\tilde\Gamma_2 \psi=0$ with $\tilde\BB=-I$ and
 \[ \label{tildeA}
 \tilde\AA := \left(\begin{array}{cc}
 -A^{-1} & -A^{-1}C^* \\ -CA^{-1} & D-CA^{-1}C^*
 \end{array} \right)\,.
 \]
It is clear that $\tilde\AA$ which equal up to the sign to
$\tilde\AA\tilde\BB^*$ is Hermitean. The reason for the
modification is that with the last boundary conditions our comparison operator $H^0$ is characterized by $\tilde\Gamma_1 \psi=0$, i.e. \mbox{$\tilde\AA^0=I$}, $\tilde\BB^0=0$. This allows us to use the result of \cite{AP05} directly (in our case it is nothing else than the usual Krein's formula) to infer that the resolvent kernel of $H_\AA$ is given by
 \begin{eqnarray} \label{krein}
 \lefteqn{G_\AA(x,x';\bx,\bx';z)= G^0(x,x';\bx,\bx';z)} \\ [.3em]
 && - G^0(x,0;\bx,\mathbf{0};z)\, \big[Q(z)-\tilde\AA\big]^{-1}
 G^0(0,x';\mathbf{0},\bx';z)\,; \nonumber
 \end{eqnarray}
the second term which is a rank sixteen operator represents the resolvent difference between the free and full Hamiltonian. It is important to note that even in the situation when the coupling is spin-independent, $\AA= {a\;\bar c \choose c\;d} \otimes \sigma_0$ and similarly for $\tilde\AA$, the Green function does not decompose; the reason is that the spin states are still coupled by the spin-orbit interaction in the plane.

\setcounter{equation}{0}
\section{Properties of $H_{\AA^{\ve}}$}

In this section we will investigate the coupling between the
halfline and the plane with a particular choice of the boundary conditions containing a parameter which allows us to control the coupling strength. Specifically, we employ the following choice of the matrices,
\begin{equation}\label{bb}
\AA^{\ve}= {\tilde{a}\;\varepsilon \choose \varepsilon
\;\tilde{d}} \otimes \sigma_0\,, \quad\varepsilon\ne 0 \quad
\mathrm{and} \quad \BB^{\ve}=- {1 \;0 \choose 0\;1} \otimes
\sigma_0\,,
\end{equation}
so the two parts of the configuration manifold are coupled and
there is no coupling between spin degrees of freedom at the
contact point. In the plane, of course, we have the spin orbit
interaction. In the decoupled case, $\varepsilon=0$, the constants $\tilde{a},\tilde{d}$ can be regarded as point interaction strengths at the halfline endpoint and in the plane, respectively, to which a physical interpretation can be given in analogy with \cite{ES97}. The off-diagonal terms in the matrix $\AA^{\ve}$ are responsible for the coupling between the halfline and the plane.

Having chosen the coupling our next task is to specify the
quantities appearing in the general formula (\ref{krein}); we get
\begin{equation*}
\left(Q(z)-\AA^{\ve}\right)^{-1}=\frac{1}{D_{\ve}(z)}\left(\begin{array}{cccc}
\Gamma^{\ve}_{1,1}(z)&0&\Gamma^{\ve}_{1,3}(z)&0\\
0&\Gamma^{\ve}_{2,2}(z)&0&\Gamma^{\ve}_{2,4}(z)\\
\Gamma^{\ve}_{3,1}(z)&0&\Gamma^{\ve}_{3,3}(z)&0\\
0&\Gamma^{\ve}_{4,2}(z)&0&\Gamma^{\ve}_{4,4}(z)
\end{array}\right)
\end{equation*}
with
\begin{eqnarray*}
D_{\ve}(z) &\!=\!&
\left\{\left(\frac{i}{\sqrt{z}}-\tilde{a}\right)
\left[G^{\mathrm{ren}}_{R,2}(z)-\tilde{d}\right]-\ve^{2}\right\}\left\{\left(\frac{i}{\sqrt{z}}-\tilde{a}\right)
\left[G^{\mathrm{ren}}_{R,1}(z)-\tilde{d}\right]-\ve^{2}\right\}
\\
\Gamma^{\ve}_{1,1}(z) &\!=\!& \left(
G^{\mathrm{ren}}_{R,1}(z)-\tilde{d}\right)
\left\{\left(\frac{i}{\sqrt{z}}-\tilde{a}\right)
\left[G^{\mathrm{ren}}_{R,2}(z)-\tilde{d}\right]-\ve^{2}\right\} \\
\Gamma^{\ve}_{1,3}(z) &\!=\!& \Gamma^{\ve}_{3,1}(z)=\ve
\left\{\left(\frac{i}{\sqrt{z}}-\tilde{a}\right)
\left[G^{\mathrm{ren}}_{R,2}(z)-\tilde{d}\right]-\ve^{2}\right\} \\
\Gamma^{\ve}_{2,2}(z) &\!=\!& \left(
G^{\mathrm{ren}}_{R,2}(z)-\tilde{d}\right)
\left\{\left(\frac{i}{\sqrt{z}}-\tilde{a}\right)
\left[G^{\mathrm{ren}}_{R,1}(z)-\tilde{d}\right]-\ve^{2}\right\} \\
\Gamma^{\ve}_{2,4}(z) &\!=\!& \Gamma^{\ve}_{4,2}(z)=\ve
\left\{\left(\frac{i}{\sqrt{z}}-\tilde{a}\right)
\left[G^{\mathrm{ren}}_{R,2}(z)-\tilde{d}\right]-\ve^{2}\right\} \\
\Gamma^{\ve}_{3,3}(z) &\!=\!&
\left(\frac{i}{\sqrt{z}}-\tilde{a}\right)
\left\{\left(\frac{i}{\sqrt{z}}-\tilde{a}\right)
\left[G^{\mathrm{ren}}_{R,2}(z)-\tilde{d}\right]-\ve^{2}\right\} \\
\Gamma^{\ve}_{4,4}(z) &\!=\!&
\left(\frac{i}{\sqrt{z}}-\tilde{a}\right)
\left\{\left(\frac{i}{\sqrt{z}}-\tilde{a}\right)
\left[G^{\mathrm{ren}}_{R,1}(z)-\tilde{d}\right]-\ve^{2}\right\}
\end{eqnarray*}

\subsection{Scattering}

The analysis of transport on a hybrid surface of mixed dimensionality follows the  scheme described in numerous papers --- see, e.g., \cite{ES87}, \cite{ES94}, \cite{ETV}, \cite{BEG03}, \cite{BG03}, and references therein.

If we have a scattering system in which the difference between the two dynamics is expressible in terms of self-adjoint extensions of a symmetric operator with finite deficiency indices it is possible to write the scattering matrix directly in terms of the Krein operator-valued function $Q(z)$ and the boundary conditions -- cf.\cite{AP} or \cite{BMN}. Using this approach --- we refer specifically to formula (1.5) from \cite{BMN} --- we arrive at
\[ \label{scattering_matrix}
S_{\AA^{\ve}}(z):=
 \left(\begin{array}{cccc}
 \frac{ \left(-\frac{\rmi}{\sqrt{z}}-\tilde{a}\right)\left(G^{\mathrm{ren}}_{R,1}(z)
 -\tilde{d}\right)-\ve^{2}}{ \left(\frac{\rmi}{\sqrt{z}}-\tilde{a}\right)\left(G^{\mathrm{ren}}_{R,1}(z)-\tilde{d}\right)-\ve^{2}} & 0 & 0&0\\
0 & \frac{ \left(-\frac{\rmi}{\sqrt{z}}-\tilde{a}\right)\left(G^{\mathrm{ren}}_{R,2}(z)-\tilde{d}\right)-\ve^{2}}
{ \left(\frac{\rmi}{\sqrt{z}}-\tilde{a}\right)\left(G^{\mathrm{ren}}_{R,2}(z)-\tilde{d}\right)-\ve^{2}} & 0 &0 \\
0 & 0 & 1 & 0
\\
0 & 0 & 0 & 1  \end{array} \right)
 \]
Putting $z=k^{2}$ we can in analogy with \cite{ETV} describe the reflection amplitude of a particle travelling along the halfline, with the ``upper'' spin component only, as follows
\begin{equation}
 \label{refcoef}
\rR(k)= \frac{ \left(-\frac{\rmi}{k}-\tilde{a}\right)\left(G^{\mathrm{ren}}_{R,1}(k^{2})-\tilde{d}\right)-\ve^{2}}
{ \left(\frac{\rmi}{k}-\tilde{a}\right)\left(G^{\mathrm{ren}}_{R,1}(k^{2})-\tilde{d}\right)-\ve^{2}}\,,
 \end{equation}
which looks like the result contained in \cite{ES07}. As briefly mentioned in the said paper, however, the magnetic case differs substantially from the non-magnetic one. The point is that the Green functions $G^{\mathrm{ren}}_{R,j}(k^{2})$ with $j=1,2$, are real-valued, and as a consequence, the scattering on the halfline is unitary, $|\rR(k)|^2=1$, in the magnetic case. From the spectral point of view one can expect formation of resonances due to the perturbation of the discrete spectrum of the spin-orbit Hamiltonian in the plane embedded in the continuous spectrum of the free Hamiltonian on the half-line as it is common in similar cases --- see again \cite{ES94}, \cite{ETV} and references therein.

\subsection{Spectral properties of $H^{\AA^{\ve}}$}

The literature devoted to the effect of weak or local
perturbations on the Landau Hamiltonian is rich. In particular, the recent paper \cite{RT} investigated a weak perturbation of the Landau Hamiltonian by a fast decaying or even compactly supported electric or/and magnetic field; the effect on the Landau levels was generically a splitting. Of course, eigenvalues which split off a Landau level appear in many situations --- recall, e.g., the classical analysis \cite{GHS} of a point interaction in a homogeneous magnetic field, possibly in presence of a potential. However, the result of \cite{RT} was more surprising: the authors discovered that a 2D axially symmetric short-range potential gives rise to an infinite number of the negative-energy levels if one takes into account the spin-orbit interaction.

Spectral properties of $H^{\AA^{\ve}}$ require an accurate
analysis, in particular, because the decoupled system has numerous embedded eigenvalues, and while Weyl's theorem guarantees stability of the essential spectrum, its character might change by the perturbation. For the sake of definiteness we suppose the magnetic field intensity is positive, $b>0$, in the chosen coordinate frame.

\newtheorem*{theo}{Theorem}\begin{theo} Assume that $\tilde{a}>0$; then the spectrum of $H_{\AA^{\ve}}$
looks as follows:
\begin{itemize}
 \setlength{\itemsep}{-4pt}
 \item[(a)] the point spectrum $\sigma_p (H_{\AA^{\ve}})
 =\sigma_p (H_{R})\cup\Sigma$, where $\Sigma$ is a finite set of
 negative eigenvalues,
\item[(b)] the continuous spectrum $\sigma_c
(H_{\AA^{\ve}})=[0,\infty)$,
\item[(c)] there are infinitely many resonances with the real parts in the gaps between the eigenvalues of $H_{R}$ and negative imaginary parts. If the coupling is weak, their distance from the embedded eigenvalues corresponding $\varepsilon=0$ is $\mathcal{O}(\varepsilon^2)$; the perturbation expansion is given by the relations (\ref{weakres}) below.
\end{itemize}
\end{theo}
\begin{proof}
The Weyl theorem guarantees the preservation of the essential
spectrum, \newline  $\sigma_{ess} (H_{\AA^{\ve}})=\sigma_{ess} (H_{0})$. Since both the operators have a common symmetric restriction with deficiency indices $(4,4)$, it follows from general principles \cite[Sec.~8.3]{We} that the negative spectrum of $\AA^{\ve}$ consists of at most sixteen eigenvalues, multiplicity taken into account.
To learn more about the the point spectrum, $\sigma_{p}
(H_{\AA^{\ve}})$, it is necessary to check directly the
singularities of the resolvent (\ref{krein}). We will analyze
explicitly one component, the other ones can be treated similarly.
We rewrite the $(3,3)$ component of the matrix (\ref{krein}) as
\begin{eqnarray}\label{inter}
\lefteqn{(G_\AA(x,x';\bx,\bx';z))_{(3,3)}} \nonumber \\ && =
G^{11}_{R}(\bx,\bx';z)-
\frac{G^{11}_{R}(\bx,0;z)\left(\frac{i}{\sqrt{z}}-\tilde{a}\right)G^{11}_{R}(0,\bx';z)}
{\left(\frac{i}{\sqrt{z}}-\tilde{a}\right)\left[G_{R,1}^{\mathrm{ren}}(z)-\tilde{d}\right]-\varepsilon^2} \nonumber \\
&&=\frac{\left(\frac{i}{\sqrt{z}}-\tilde{a}\right)\left[G_{R,1}^{\mathrm{ren}}(z)G^{11}_{R}(\bx,\bx';z)
-G^{11}_{R}(\bx,\mathbf{0}';z)G^{11}_{R}(\mathbf{0},\bx';z)\right]}{\left(\frac{i}{\sqrt{z}}
-\tilde{a}\right)\left[G_{R,1}^{\mathrm{ren}}(z)-\tilde{d}\right]-\varepsilon^2} \nonumber \\
&&\quad-\frac{\left[\left(\frac{i}{\sqrt{z}}-\tilde{a}\right)\tilde{d}
+\varepsilon^2\right]G^{11}_{R}(\bx,\bx';z)}{\left(\frac{i}{\sqrt{z}}
-\tilde{a}\right)\left[G_{R,1}^{\mathrm{ren}}(z)-\tilde{d}\right]-\varepsilon^2}
\end{eqnarray}
where $G^{11}_{R}(\bx,\bx';z)$ equals
\begin{equation*}
\left( \frac{1}{2}+\frac{\beta_{R}
-\varkappa_{R}}{2\eta_{R}}\right)G_0\big(\bx,\bx';\zeta_{R}^-( b)\big)
+\left( \frac{1}{2}-\frac{\beta_{R}-\varkappa_{R}}{2\eta_{R}}\right)G_0\big(\bx,\bx';\zeta_{R}^+( b)\big)
\end{equation*}
Expanding the digamma function in $G_0\big(\bx,\bx';z\big)$ we get
\begin{equation*}
\begin{split}
&G_0\big(\bx,\bx';z\big)=\frac{1}{4\pi} \Theta(\bx,\bx')\left[\Phi
\Big(\frac{1}{2}-\frac{z}{2|b|}\,,1;\frac{|b|}{2}(\bx-\bx')^2\Big)\,\log\left(\frac{|b|}{2}(\bx-\bx')^{2}\right)\right.\\
&+\left.\sum_{r=0}^{\infty}\frac{\left(\frac{1}{2}-\frac{z}{2\,b}\right)_{r}}{r!}\left(-4\pi\,U(z-2|b|r)
-\log\frac{|b|}{2}+2\Psi(1)-2\Psi(1+r)\right)x^{r}\right]
\end{split}
\end{equation*}
where
\[\nonumber\Theta(\bx,\bx'):=\exp\Big(\frac{\rmi
b}{2}(\bx\wedge\bx')-\frac{|b|}{4}(\bx-\bx')^2\Big)\]
and
\[\nonumber\Phi(a,c,x)
=\sum_{r=0}^{\infty}\frac{(a)_{r}}{(c)_{r}r!}x^{r} \quad
\mathrm{with} \quad (z)_{r} =\frac{\Gamma(z+r)}{\Gamma(z)}\,.\]
It is easy to check that the first term of the last expression at the \emph{rhs} of (\ref{inter}) is still singular for
$\zeta_{R}^\pm( b)=n$ and these singularities are the same as those of the decoupled Hamiltonian. In a similar way we can treat the other components of the resolvent (\ref{krein}) concluding that the Rashba energy levels remain to be embedded eigenvalues of infinite multiplicity in the coupled case. Using this fact and the above mentioned preservation of the essential spectrum it is possible to determine  the continuous part of the spectrum,
\begin{equation*}
\sigma_{c} (H_{\AA^{\ve}})=\sigma_{c} (H_{0})
\end{equation*}

The resolvent (\ref{krein}) can have naturally other singularities coming from zeroes of $D_{\ve}(z)$. In the decoupled case, $\varepsilon=0$, the point interaction on the halfline has one bound state with the negative energy $E_{0}=
-\frac{1}{\tilde{a}^2}$, since we have assumed $\tilde{a}>0$. Let us analyze the effect of the coupling, $\varepsilon\neq 0$, observing the component $(j,j)$ of (\ref{krein}) with $j=1,2$,
\begin{eqnarray*}
\left(G_{\AA^{\ve}}(x,x';\bx,\bx';z)\right)_{(j,j)}
&\!=\!&\frac{\rmi}{\sqrt{z}}\, \cos\left( \sqrt{z}x_<\right)
\re^{-\rmi\sqrt{z}x_>} \\ &&
+\left(\frac{\rmi}{\sqrt{z}}\right)^2\frac{\re^{-\rmi\sqrt{z}x}\,
\re^{-\rmi\sqrt{z}x'}}{\frac{\rmi}{\sqrt{z}}
-\tilde{a}-\frac{\varepsilon^2}{G_{R,j}^{\mathrm{ren}}(z)-\tilde{d}}}
\end{eqnarray*}
The only singularity in the vicinity of $E_{0}$  is obtained from the equation
\begin{equation*}
\left(\frac{i}{\sqrt{z}}-\tilde{a}\right) \left[G^{\mathrm{ren}}_{R,j}(z)-\tilde{d}\right]-\ve^{2}=0\quad j=1,2
\end{equation*}
The Rashba eigenvalues of the unperturbed Hamiltonian,
corresponding to $\varepsilon=0$, $\tilde{d}=0$, are all real,
positive and infinitely degenerate, and since the perturbation is of finite rank in the resolvent sense, continuity of the
singularities w.r.t. $\varepsilon$ is guaranteed.

For a start let us discuss qualitatively the solutions of the equation $D_{\ve}(z)=0$ around Rashba eigenvalues. For the
decoupled case, $\varepsilon=0$, and $\tilde{d}\neq 0$ we a point interaction in the plane with spin-orbit interaction and the effect of this perturbation is analogous to that of \cite{GHS}. Suppose that $\tilde{d}$ is sufficiently large to keep all the eigenvalues real and positive  --- we recall that the lowest one goes to $-\infty$ as $\tilde{d}\to -\infty$ --- which means in the language of two-dimensional point interactions that the perturbation is sufficiently weak. Switching then the coupling on, $\ve\ne0$, the condition $D_{\ve}(z)=0$ with the solution $z=\lambda\in\mathbb{R}_+$ can be then rewritten as
\[\label{qual}
\sqrt{\lambda}=i\frac{G^{\mathrm{ren}}_{R,j}(\lambda)-\tilde{d}}
{\tilde{a}\left[G^{\mathrm{ren}}_{R,j}(\lambda)-\tilde{d}\right]+\ve^{2}}\,,\quad
j=1,2
\]
and we notice that this equation cannot be solved by any real
positive $\lambda$ because on the left-hand side we have always a real positive number while the right-hand side is always purely imaginary. This is due to the fact that the function
$G^{\mathrm{ren}}_{R,j}(\lambda)$ takes real values for every $\lambda>0$ and $j=1,2$. If there is a solution to this equation it must be therefore outside the real axis.

On the other hand, for a sufficiently strong point interaction in the couple case, $\varepsilon=0$, in other words for $\tilde{d}$ sufficiently large negative there may be negative eigenvalues. The equation for singularities in this case  can be written as
\begin{equation}
{G^{\mathrm{ren}}_{R,j}(z)}={\tilde{d}}
\end{equation}
From the properties of the digamma function we see that the
function $G^{\mathrm{ren}}_{R,j}(z)$ is \newline monotonously decreasing around the origin, it is positive in a neigborhood when it approaches the first Rashba level and has a zero on the negative halfline. Consequently, for $\tilde{d}$ sufficiently large negative there is a negative solution. Switching now again the coupling in, $\ve\ne 0$, an taking $z=-\lambda$ with $\lambda>0$ we see from (\ref{qual}) that such a solution is still a negative eigenvalue, at least for $|\ve|$ small enough.

Let us now look more closely how the singularities of the
resolvent behave in the weak coupling regime. We are going to find the corresponding series expansion in $\varepsilon$ using a recursive procedure. To begin with we rewrite the equation
$D_{\ve}(z)=0$ as
\begin{equation}
f(z)=\frac{1}{G^{\mathrm{ren}}_{R,j}(z)}
\label{geneve}
\end{equation}
where we have introduced
\begin{equation}\label{pos}
f(z):=\frac{1}{\tilde{d}-\frac{\ve^{2}\sqrt{z}}{\tilde{a}\sqrt{z}-i}}
=\frac{(\tilde{a}\tilde{d}-\ve^{2})\tilde{a}\sqrt{z}+\tilde{d}}{(\tilde{a}\tilde{d}
-\ve^{2})^{2}z+\tilde{d}^{2}}+\frac{i\ve^{2}\sqrt{z}}{(\tilde{a}\tilde{d}-\ve^{2})^{2}z+\tilde{d}^{2}}
\end{equation}
To simplify the recursive procedure we adopt the following
notation,
\begin{equation}
\label{rashbalevel}
\left.\begin{array}{c}E^{0}_{1,n}=2nb-2\varkappa_r\sqrt{\beta^{2}_{R}+2nb}
\\ [.3em] E^{0}_{2,n}=2nb+2\varkappa_r\sqrt{\beta^{2}_{R}+2nb}
\\ [.3em] E^{0}_{3,n} =2(n+1)b-2\varkappa_r\sqrt{\beta^{2}_{R}+2(n+1)b}
\\ [.3em] E^{0}_{4,n}=2(n+1)b+2\varkappa_r\sqrt{\beta^{2}_{R}+2(n+1)b}\end{array}\right.
\end{equation}
and we will look for the fixed points of
\begin{equation}\label{rec}
\frac{1}{G^{\mathrm{ren}}_{R,j}(z^{(k)})}\simeq
\rho(E^0_{l,n})(z^{(k)}-z^{(0)})=f(z^{(k-1)})\,,\quad
l=1,2,3,4,\,j=1,2\,,
\end{equation}
where
\begin{equation}
\rho(E_{l,n}^{0})=\left.\frac{\partial}{\partial
z}\left(\frac{1}{G_{R,1}^{\mathrm{ren}}(z)}\right)\right|_{z=E_{l,n}^{0}}\,,\quad
l=1,2\,,
\end{equation}
\begin{equation}
\rho(E_{l,n}^{0})=\left.\frac{\partial}{\partial
z}\left(\frac{1}{G_{R,2}^{\mathrm{ren}}(z)}\right)\right|_{z=E_{l,n}^{0}}\,,\quad
l=3,4\,,
\end{equation}
\begin{equation}\label{deriv}
\begin{split}\rho(E_{2,n}^{0})= \rho(E_{3,n-1}^{0})
=-\frac{8\pi}{b}\frac{\sqrt{2nb+\beta^{2}_{R}}}{\beta_{R}^{2}+\sqrt{2nb+\beta_{R}^{2}}}\,,
\\\rho(E_{1,n}^{0})= \rho(E_{4,n-1}^{0})=-\frac{8\pi}{b}\frac{\sqrt{2nb+\beta^{2}_{R}}}{\beta_{R}^{2}-\sqrt{2nb+\beta_{R}^{2}}} \,. \end{split}
\end{equation}
In this way we can determine the position of the singularities in the leading order in $\ve$. Starting with $z^{(0)}=E_{l,1}^{0}+\frac{1}{\rho(E_{l,n}^{0})}
\frac{1}{\tilde{d}}$
and $l=1,3$ we have two negative real solutions if
$\tilde{d}\leqslant\min\{C_{1},C_{2}\}$ with
$C_{1}=(E_{1,1}^{0}\rho(E_{1,1}^{0}))^{-1}$ and
$C_{2}=(E_{3,1}^{0}\rho(E_{3,1}^{0}))^{-1}$.

On the other hand, for $\tilde{d}$ sufficiently large the starting point of the recursive procedure will be a positive eigenvalue $z^{(0)}=E_{l,n}^{0}+\frac{1}{\rho(E_{l,n}^{0})}
\frac{1}{\tilde{d}}$ and the fixed points of the recursive
procedure, from (\ref{rec}) and (\ref{deriv}), will have a
positive real part and a negative imaginary part; in leading order in $\varepsilon$ the real and imaginary parts of the resonances are
\begin{equation}
\label{weakres}
\begin{split}
&\mathcal{R} (E_{l,n}^{res})=E_{l,n}^{0}+\frac{1}{\rho(E_{l,n}^{0})}\left(\frac{1}{\tilde{d}}+\frac{\tilde{a}E_{l,n}^{0}}{\tilde{d}^{2}(1+\tilde{a}^{2}E_{l,n}^{0})}\varepsilon^{2}\right)+\mathcal{O}(\varepsilon^{3})\\
&\mathcal{I}
(E_{l,n}^{res})=\frac{1}{\rho(E_{l,n}^{0})}\frac{\sqrt{E_{l,n}^{0}}\,\varepsilon^{2}}{\tilde{d}^{2}(1+\tilde{a}^{2}E_{l,n}^{0})}+\mathcal{O}(\varepsilon^{3})
\end{split}
\end{equation}
\noindent From the first of (\ref{weakres}), putting $\ve=0$, it is obtained the displacement of the Rashba energy levels due to the point interaction in the plane. In the coupling case, for $\ve\neq 0$ the energy eigenvalue will migrate in a resonance with negative imaginary part.

The convergence of the sequence in a ball of radius $\ve^2$
centered at $z^{(0)}=E_{l,n}^{0} +\frac{1}{\rho(E_{l,n}^{0})}
\frac{1}{\tilde{d}}$ can be proved using the following estimates,
$$
\|z^{(1)}-z^{(0)}\|=\frac{1}{\rho(E_{l,n}^{0})}\|f(z^{(0)})\|<
C_1 \varepsilon^2\,,
$$
$$
\|z^{(k+1)}-z^{(k)}\|\leqslant C_2\, \varepsilon^2
\|z^{(k)}-z^{(k-1)}\|\,,
$$
where the constants $C_j$ are independent of $\varepsilon$.
\end{proof}

\begin{remark}
{\rm The assumption $\tilde{a}>0$ was used to ensure that there are proper eigenvalues in the decoupled case on the halfline. In a similar way one can treat the case $\tilde{a}<0$ where for $\varepsilon=0$ there is a singularity for the decoupled resolvent but it is on the second Riemann sheet, i.e. an antibound state. One can check that for $\varepsilon\neq 0$ this point will remain on the real negative axis on the second Riemann sheet.}
\end{remark}

\medskip

In order to illustrate the above conclusions, one can analyze in an example the effect of the coupling on the positive eigenvalues coming from the Rashba energy levels in the plane. For instance, let us consider the following eigenvalue in the decoupled case:
\[E_{1,2}^{0}=E_{3,1}^{0}=2 b -2\varkappa_{R}\sqrt{\beta^{2}_{R}+2 b}\]
We fix the value of $\tilde{d}=1$ and study numerically the effect of the coupling $\ve\neq 0$. In the following pictures the real and the imaginary parts of the resonances arising from the same Rashba eigenvalue $E_{1,2}^{0}=E_{3,1}^{0}$ are shown.

\begin{figure}[H]\centering
\includegraphics[width=10 cm]{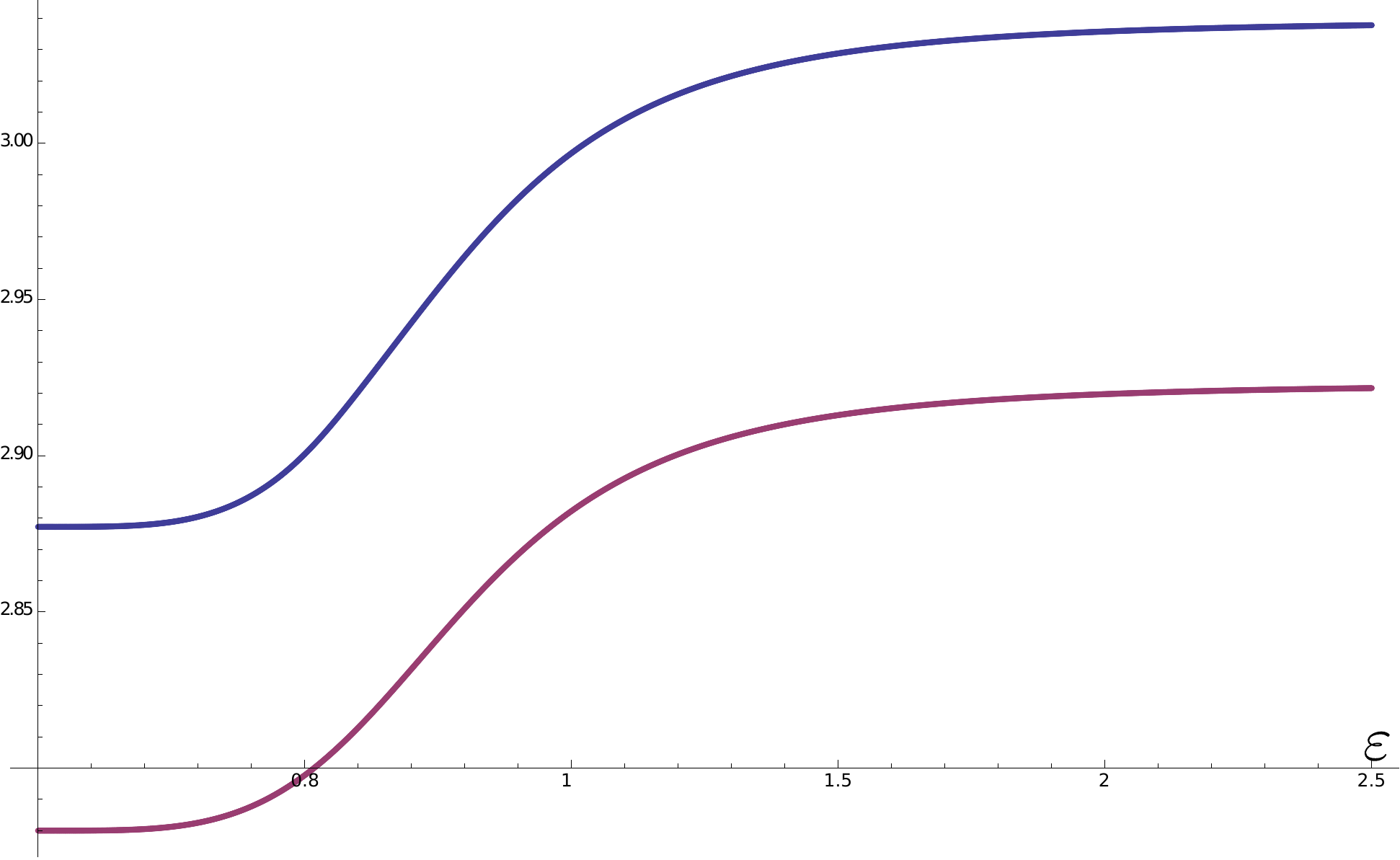}\label{fig6}
\caption{In blue the real part of the resonance is plotted as a
function of $\varepsilon$, coming from $E_{1,2}^{0}$ (color
online). In red (the lower curve) the real part of the resonance
coming from $E_{3,1}^{0}$ with $\tilde{d}=1$, $\tilde{b}=1$,
$\tilde{a}=0$ and $\chi/b=0.1$ is shown.}
\end{figure}
\begin{figure}[H]\centering
\includegraphics[width=10 cm]{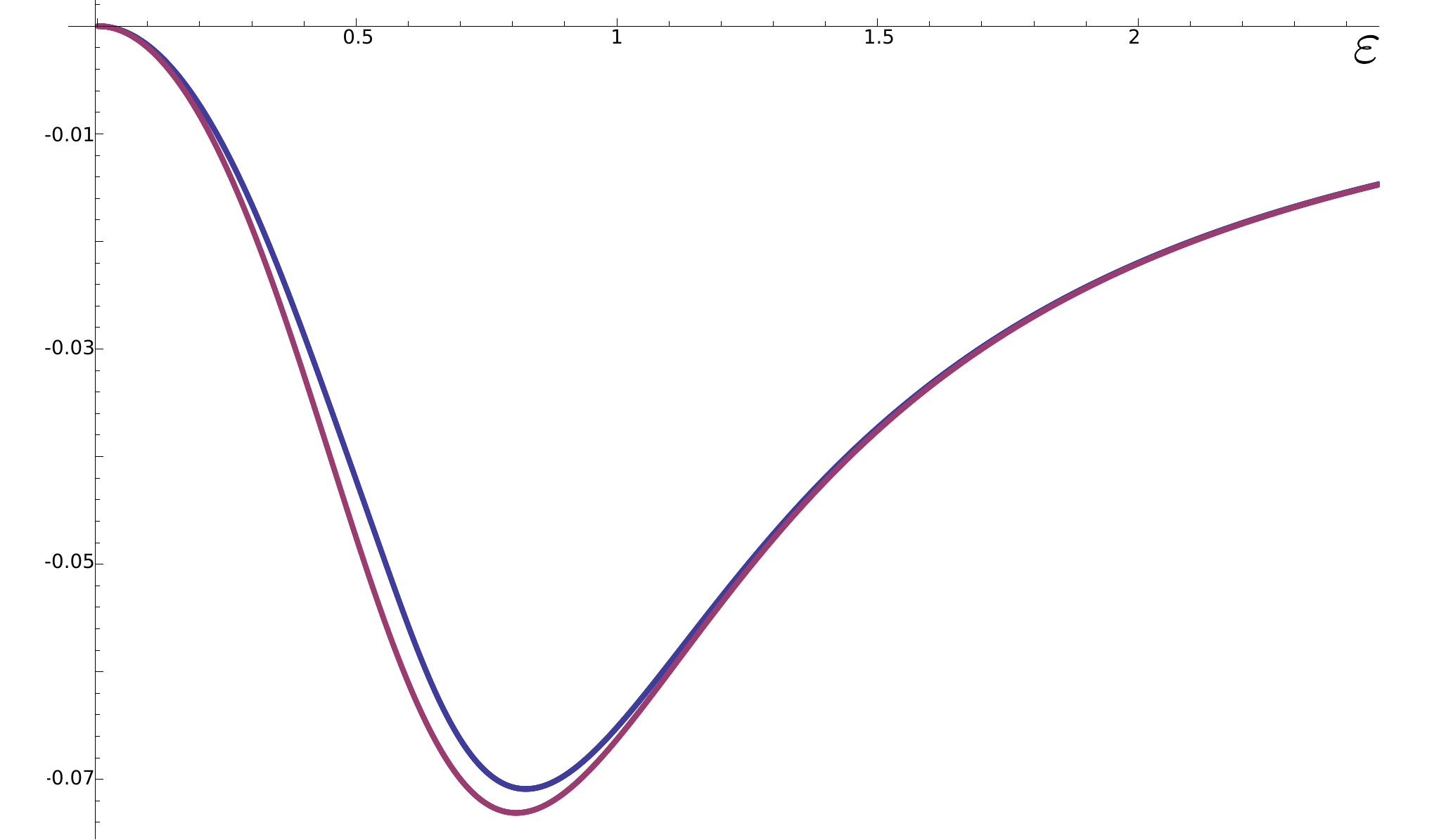}\label{fig7}
\caption{The imaginary part of the same resonances as above.}
\end{figure}

As it often happens in such resonance problems --- see, e.g.,
 \cite{ES94} --- for large values of the coupling constant $\ve$ the resonance can approach an eigenvalue again. In
our case it is clear that \ref{bc} reduces in the limit
$\ve\rightarrow\infty$ to Dirichlet boundary conditions.

\section*{Acknowledgments}
The authors are grateful to the referees for suggestions which helped to improve the manuscript. The research was supported by the Czech Ministry of Education, Youth and Sports within the project LC06002.

\end{document}